\documentclass[reqno,12pt]{amsart}

\input texdraw

\newtheorem{theorem}{Theorem}[section]
\newtheorem{lemma}[theorem]{Lemma}
\newtheorem{example}[theorem]{Example}
\newtheorem{proposition}[theorem]{Proposition}

\theoremstyle{remark}
\newtheorem*{remark}{Remark}

\newcommand{\ds}{\displaystyle}
\newcommand{\new}{\vskip0.4cm\noindent}
\newcommand{\bi}{\bibitem}

\newcounter{smalllist}

\begin{document}

\title[Non-Random Perturbations of the Anderson Hamiltonian]{Non-Random Perturbations of the
Anderson Hamiltonian in the 1-D case.}

\author[J. Holt, S. Molchanov, B. Vainberg]{J. Holt$^1$, S. Molchanov$^2$, B. Vainberg$^3$}

\thanks{$^{1}$ University of South Carolina Lancaster, Lancaster, SC 29721, USA.
email: jholt@mailbox.sc.edu}
\thanks{$^{2}$ University of North Carolina Charlotte, Charlotte, NC 28223, USA.
email: smolchan@uncc.edu}
\thanks{$^{3}$ University of North Carolina Charlotte, Charlotte, NC 28223, USA.
email: brvainbe@uncc.edu}

\begin{abstract} Recently (see \cite{MV}), two of the authors applied the Lieb method to the study of the negative
spectrum for particular operators of the form $H=H_0-W$.  Here, $H_0$ is the generator of the
positive stochastic (or sub-stochastic) semigroup, $W(x) \geq 0$ and $W(x) \to 0$ as $x \to
\infty$ on some phase space $X$. They used the general results in several ``exotic" situations,
among them the Anderson Hamiltonian $H_0$.  In the 1-d case, the subject of the present paper,
we will prove similar but more precise results.

\new \textsc{Keywords}: Anderson Hamiltonian, Sch\"{o}dinger operator, negative eigenvalues
\end{abstract}

\maketitle

\section{Introduction} \label{s1}

\new  The basic Anderson Hamiltonian $H_0$ has the form $H_0 \psi =-\psi''+V_\omega \psi$ where $V_\omega(x)$ is
a nonnegative random potential. We'll consider here $V_\omega$ in the following way.  Let $x_n \in
\mathbb{R}$ be an increasing sequence and consider a partition of $\mathbb{R}$ into intervals
$I_k$ defined by $I_k =\{x : -l \leq x-x_k \leq l\}$ for $k \in \mathbb{Z}$ and positive $l$.  Let
$L_k=x_{k+1}-x_k-2l$ represent the distance between consecutive intervals $I_k$.  We'll suppose
that $L_k(\omega)\geq 0, \omega \in (\Omega,\mathcal{F},P)$ are independent and identically
distributed random variables which are unbounded from above, with some density $p(x)$ and finite
expectation $$\mathbb{E}[L_k]=\int^\infty_0 xp(x) \ dx.$$  Define the random potential $V_\omega$
by
$$V_\omega(x)=\ds \sum_{k \in \mathbb{Z}} h \mathbb{I}_k(x)$$ where $\mathbb{I}_k$ is the indicator
function of the interval $I_k$ and $h>0$. It is easy to see that $H_0 \geq 0$ and that each
realization $P$-a.s contains arbitrarily long intervals where $V_\omega=0$. This implies that
Sp$(H_0)=[0,\infty)$ $P$-a.s. Let $W(x)$ be a continuous function on $\mathbb{R}$ such that $W(x)
\geq 0$ and $W(x) \to 0$ as $|x| \to \infty$. Then the negative spectrum $\mbox{Sp}(H)$ of the
non-random perturbation of $H_0$ defined by $H=H_0-W$ is discrete with a possible accumulation
point only at 0, and the random variable $\mathcal{N}_0(W)=\#\{\lambda_i(\omega) <0\}$ can be
either finite or infinite. Here, $\{\lambda_i\}=\mbox{Sp}(H) \cap (-\infty,0)$ so that
$\mathcal{N}_0(W)$ represents the number of negative eigenvalues. Note that the event $\{\omega
\in \Omega : \mathcal{N}_0(W)=\infty\}$ belongs to the tail $\sigma$-algebra of the potential
$V_\omega$ and therefore the Kolmogorov zero-one law implies that either
$P\{\mathcal{N}_0(W)<\infty\}=1$, or $P\{\mathcal{N}_0(W)<\infty\}=0$. Our goal in this work is to
find conditions on $W$ for which $P(\{\mathcal{N}_0<\infty\})=1$ and
$\mathbb{E}[\mathcal{N}_0]<\infty$.  In particular, we'll establish the following fact for a
Bernoulli potential as a consequence of our general results on the wider class of Kronig-Penney
type potentials. Such results are related to the classical renewal process in reliability theory.

\begin{theorem} \label{exponentialtails} For $h>0$ and any $0<p<1$ let $V_\omega=h\sum_{k \in \mathbb{Z}}
\epsilon_k(\omega) \mathcal{X}_k$ where $\{\epsilon_k(\omega)\}$ are independent and
identically distributed with $P(\epsilon_k=1)=p$ and $P(\epsilon_k=0)=1-p$. Here
$\mathcal{X}_k$ is the indicator of the interval $[k-1/2,k+1/2)$. Then if $W(x) < \ds
C_p/\ln^2(|x|+1)$ for large enough $x$, then $P\{\mathcal{N}_0(W)<\infty\}=1$, while $P\{\mathcal{N}_0(W)<\infty\}=0$
if $W(x)> \ds C_p/\ln^2(|x|+1)$.  Here, $C_p$ is the constant
$C_p=\log^2_{1/p} \pi^2$.
\end{theorem}

\new Our results here are similar to the $d$-dimensional results obtained in \cite{MV}.  However, whereas we were able to
establish a distinct borderline between finitely many and infinitely many negative eigenvalues in the one dimensional case by finding
$C_p$ exactly, the result in \cite{MV} only shows the following:

\begin{theorem} Let $V_\omega=h\sum_{k \in \mathbb{Z}^d} \epsilon_k(\omega) \mathcal{X}_k$ where
$\mathcal{X}_k$ is the indicator of the cube $\{(x_1,\ldots, x_d): -1/2 \leq x_i-k_i<1/2\}$. There exist positive
constants $C_-<C_+$ so that if $0 \leq W(x) \leq \ds
C_-/\ln^2(|x|+2)$
then $P\{\mathcal{N}_0(W)<\infty\}=1$ while if
$W(x) \geq \ds
C_+/\ln^2(|x|+2)$ then $P\{\mathcal{N}_0(W)=\infty\}=1$.
\end{theorem}

\begin{remark} Estimations of the constants $C_-$ and $C_+$ in the above theorem are not known as
they are related to difficult problems in percolation theory.
\end{remark}

\begin{remark} Let's note that $V_\omega(x)$ can be viewed as an an asymptotically stationary ``renewal"
process in which $L_k$ can be thought of as the time between successive occurrences of the event
$V_\omega(x)=h$.
\end{remark}

\new Additionally we formulate and prove two more results showing
a clear relationship between the tails $P(L_k>x)$ and the function $W$,
and which also demonstrate a distinct borderline between finite and infinitely many eigenvalues.
The first of these results when $L_k$ have very ``light tails", that is $P(L_k>x)
\sim \exp(-\eta x^\alpha/\alpha)$ as $x \to \infty$ where $\eta,\alpha>0$, is actually a generalization
of Theorem \ref{exponentialtails}. In fact, Theorem \ref{exponentialtails} can be
obtained by considering the case $\alpha=1$.

\begin{theorem} \label{lighttails} Suppose $P\{L_k > x\} \sim e^{-\eta x^\alpha/\alpha}$ as
$x \to \infty$ for some $\eta,\alpha>0$. Let $C_{\eta,\alpha}=(\eta \alpha^{-1})^{2/\alpha}
\pi^2.$ If for large $k$ $w_k <\ds C_{\eta,\alpha}/\ln^{2/\alpha} k$, then
$\mathcal{N}_0(W)<\infty$ $P$-a.s. On the other hand, if for large $k$, $w_k> \ds
C_{\eta,\alpha}/\ln^{2/\alpha} k$, then $\mathcal{N}_0(W)=\infty$ $P$-a.s.
\end{theorem}

\begin{remark} It shows that the borderline between $P$-a.s finiteness
of $\mathcal{N}_0(W)$ goes along the line
$$w_k \sim \frac{C_{\eta,\alpha}}{\ln^{2/\alpha} k}.$$
\end{remark}

\begin{theorem} Suppose $P(L_k>x) \sim c/x^\alpha$ as $x\to \infty$ with $\alpha>1$.
If for large $k$, $w_k<t^{-\beta}$ with $\beta>2/\alpha$, then $\mathcal{N}_0(W)<\infty$ $P$-a.s. and
$\mathbb{E}[\mathcal{N}_0(W)]<\infty$.  If on the other hand, $\beta<2/\alpha$, then $\mathcal{N}_0(W)=\infty$ $P$-a.s. and
$\mathbb{E}[\mathcal{N}_0(W)]=\infty$. It shows that the borderline between $P$-a.s finiteness
of $\mathcal{N}_0(W)$ and the expectation $\mathbb{E}[\mathcal{N}_0(W)]$ goes along the line
$$w_k \sim \frac{1}{k^{2/\alpha}}.$$
\end{theorem}

\begin{remark}
Both of these asymptotic equivalences are corollaries of more general results established in
Theorems 4.1 and 4.2.
\end{remark}

\new It is helpful to first consider the Anderson Hamiltonian $H_0$ on the half axis
defined by $$H_0 \psi =-\psi''+V_\omega \psi, \hskip1cm \psi(0)=0$$ and the case where $h=\infty$
and $W$ constant on the intervals $I_k=[x_k+l,x_{k+1}-l]$. Let $W(x)=w_k$ for $x \in I_k$. With
$h=\infty$, we must interpret $H$ as a direct sum of operators $H_k=-d^2/dx^2-w_k \mathbb{I}_k$ on
a system of decoupled wells $[x_k+l,x_{k+1}-l]$ with Dirichlet boundary conditions at the edges.
For each $k$, the spectrum of $H_k$ is discrete with eigenvalues
$$\lambda_{k,i}=\frac{\pi^2 i^2}{L^2_k}-w_k$$ and $\mathcal{N}_{0,k}(W)=\#\{\lambda_{k,i}<0\}=\ds
\left\lfloor \frac{\sqrt{w_k} L_k}{\pi}\right\rfloor$. The total number of negative eigenvalues
for $H$ is then $$\mathcal{N}_0(W)=\sum^\infty_{k=1} \left\lfloor \frac{\sqrt{w_k} L_k}{\pi}
\right\rfloor.$$  This series of independent random variables is finite if and only if the series
contains finitely many non-zero terms.  But,
$$P\left\{\left\lfloor \frac{\sqrt{w_k} L_k}{\pi}\right\rfloor
\geq 1 \right\}=P\left\{ L_k > \frac{\pi}{\sqrt{w_k}}\right\}= \int^\infty_{\pi/\sqrt{w_k}} p(x)
dx.$$  Due to the Borel-Cantelli lemmas, we have the following proposition:

\begin{proposition} $\mathcal{N}_0(W)< \infty$ $P$-a.s. if and only if $$\sum^\infty_{k=1}
\int^\infty_{\pi/\sqrt{w_k}} p(x) dx =\sum^\infty_{k=1} P(L_1>\pi/\sqrt{w_k})<\infty.$$
\end{proposition}

\new Considering the expectation for $\mathcal{N}_0(W)$, it is clear that
\begin{equation} \int^\infty_{\pi/\sqrt{w_k}} \left(\frac{\sqrt{w_k}}{\pi}x-1\right)p(x) dx \leq
\mathbb{E}[\mathcal{N}_{0,k}] \leq \int^\infty_{\pi/\sqrt{w_k}} \frac{\sqrt{w_k}}{\pi}x p(x) dx
\label{expect}
\end{equation}
since $$\mathbb{E}[\mathcal{N}_{0,k}]=\ds \sum_{j \geq 1} P(\mathcal{N}_{0,k} \geq j)=\sum_{j\geq
1} P\left(L_k > \frac{j\sqrt{w_k}}{\pi}\right).$$  With $F_k(x)=P(L_k>x)$ and the condition of
finiteness of the expectation of $L_k$, it follows by integration by parts that
$$\int^\infty_{\pi/\sqrt{w_k}} xp(x) \; dx=
\int^\infty_{\pi/\sqrt{w_k}} F(x) \;
dx+\frac{\pi}{\sqrt{w_k}}F\left(\frac{\pi}{\sqrt{w_k}}\right)$$ so that
$$\frac{\sqrt{w_k}}{\pi}\int^\infty_{\pi/\sqrt{w_k}} F_k(x) \; dx \leq
\mathbb{E}[\mathcal{N}_{0,k}] \leq \frac{\sqrt{w_k}}{\pi}\int^\infty_{\pi/\sqrt{w_k}} F_k(x) \;
dx+F\left(\frac{\pi}{\sqrt{w_k}}\right).$$ It gives immediately necessary and sufficient
conditions for the finiteness of the expectation.

\new Before considering the general case, we present a few examples to demonstrate how the borderline behavior
between finitely many and infinitely many negative eigenvalues depend on the decay of $W$ and the
bump separation lengths $L_k$.

\begin{example}  Let's consider the situation when $L_k\geq 0$ are distributed with exponential tails,
that is $P\{L_k > x\}=e^{-\eta x}$ for some $\eta>0$. If for large $k$ $w_k <\ds
\frac{\eta^2\pi^2}{\ln^2 k}$, then $\mathcal{N}_0(W)<\infty$ $P$-a.s. On the other hand, if for
large $k$, $w_k> \ds \frac{\eta^2\pi^2}{\ln^2 k}$, then $\mathcal{N}_0(W)=\infty$ $P$-a.s. It
shows that the borderline between $P$-a.s finiteness of $\mathcal{N}_0(W)$ goes along the line
$$w_k \sim \frac{\eta^2\pi^2}{\ln^2 k}.$$
\end{example}

\new

\begin{remark} Actually, we can say even more in the previous example.
Again, assuming that $P\{ L_k \geq x\}=e^{-\eta x}$, $\mathbb{E}[L_k]=1$ and $P( L_k \geq \pi
w^{-1/2}_k)=\exp\left(-\eta\pi w^{-1/2}_k\right)$. These exponential tails are given by the
density function $p(x)=\eta e^{-\eta x}$. Therefore by (\ref{expect})
\begin{equation}\ds \frac{\sqrt{w_k}}{\eta \pi}e^{-\eta\pi/\sqrt{w_k}}
\leq
\mathbb{E}[\mathcal{N}_{0,k}] \leq\left(\frac{\sqrt{w_k}}{\eta \pi}+1\right) e^{-\eta\pi/\sqrt{w_k}}
\label{expect2}
\end{equation} as $w_k \to 0$.
So, for small $\delta>0$ and $$w_k=\frac{\eta^2\pi^2(1-\delta)^2}{\ln^2 k}$$ it follows that
$\mathbb{E}[\mathcal{N}_{0,k}] =O\left(n^{-1-\beta}\right)$ where $1+\beta=(1-\delta)^{-1}$.
Consequently $\mathcal{N}_0(W)=\ds \sum_k \mathbb{E}[\mathcal{N}_{0,k}]<\infty$.  If instead
$$w_k=\frac{\eta^2 \pi^2 (1+\delta)^2}{\ln^2 k}$$ then by (\ref{expect2})
$$\mathbb{E}[\mathcal{N}_{0,k}] \geq \frac{C_\delta}{k^{1-\beta}\ln k}$$
with $1-\beta=(1+\delta)^{-1}$ and consequently $\mathbb{E}[\mathcal{N}_0(W)]=\infty$.
This shows that the decay in $W$ and the borderline between $\mathbb{E}[\mathcal{N}_0(W)]<\infty$
and $\mathbb{E}[\mathcal{N}_0(W)]=\infty$ is also along $$W(x) \sim \frac{\eta^2\pi^2}{\ln^2 x}.$$
\end{remark}

\begin{example} Let's consider the heavy tailed case where $P(L_k>x) \sim c/x^\alpha$ as $x\to
\infty$ with $\alpha>2$.  If for large $k$, $w_k<t^{-\beta}$ with $\beta>2/\alpha$, then
$\mathcal{N}_0(W)<\infty$ $P$-a.s. and $\mathbb{E}[\mathcal{N}_0(W)]<\infty$.  If on the other
hand, $\beta<2/\alpha$, then $\mathcal{N}_0(W)=\infty$ $P$-a.s. and
$\mathbb{E}[\mathcal{N}_0(W)]=\infty$. It shows that the borderline between $P$-a.s finiteness of
$\mathcal{N}_0(W)$ and the expectation $\mathbb{E}[\mathcal{N}_0(W)]$ goes along the line
$$w_k \sim \frac{1}{k^{2/(\alpha-1)}}.$$
\end{example}

\section{Description of the Model in the General Case}

\noindent One can split the real line $\mathbb{R}$ into two half axes, and introduce a Dirichlet
boundary condition at the point of splitting, and reduce the initial spectral problem on
$\mathbb{R}$ to the union of two similar problems on the half axes. Consider the 1-D
Schr\"{o}dinger operator $H_0\psi=-\psi''+V_\omega(x)\psi$ with boundary condition $\psi(0)=0$ for
$x \geq 0$. The random potential $V_\omega$ is defined by

$$V_\omega(x)=\ds \sum_{k \in \mathbb{Z}} h \mathbb{I}_k(x,\omega)$$
where $h>0$, $\mathbb{I}_k$ is the indicator of the interval $I_k=[x_k-l,x_k+l]$ and $x_k>0$ is a
random increasing sequence of real numbers. A typical realization of $V_\omega$ is given in Figure
1 below, and consists of bumps of constant height $h>0$ and width $2l>0$, separated by wells of
random length $L_k=x_{k+1}-x_k-2l$ where the potential $V_\omega=0$.  We will assume that $L_k$
are independent and identically distributed with some density $p(x)$ and finite expectation
$$\mathbb{E}[L_k]=\int^\infty_0 xp(x) \; dx<\infty$$ and that for all $t>0$, $P(L_k>t)>0$. Let's
note that the condition $\mathbb{E}[L_k]<\infty$ implies that $V(x,\omega)$ is stationary near the
point $\infty$.

\vfill
\centertexdraw{

    \move(-0.25 0) \avec(-0.25 1.25)
    \move(-0.25 0) \avec(4.5 0)

    \move(0.5 0.5) \lvec(1.0 0.5)
    \move(1.6 0.5) \lvec(2.1 0.5)
    \move(3.1 0.5) \lvec(3.6 0.5)

    \move(0.5 0) \lpatt(0.037 0.05) \lvec(0.5 0.5)
    \move(1.0 0) \lvec(1.0 0.5)
    \move(1.6 0) \lvec(1.6 0.5)
    \move(2.1 0) \lvec(2.1 0.5)
    \move(3.1 0) \lvec(3.1 0.5)
    \move(3.6 0) \lvec(3.6 0.5)
    \htext(0.7 0.55){$\scriptstyle h$}
    \htext(1.8 0.55){$\scriptstyle h$}
    \htext(3.3 0.55){$\scriptstyle h$}
    \htext(0.35 -0.15){$\scriptstyle x_1-l$}
    \htext(0.85 -0.15){$\scriptstyle x_1+l$}
    \htext(1.45 -0.15){$\scriptstyle x_2-l$}
    \htext(1.95 -0.15){$\scriptstyle x_2+l$}
    \htext(2.95 -0.15){$\scriptstyle x_3-l$}
    \htext(3.45 -0.15){$\scriptstyle x_3+l$}

    \htext(0.05 0.02){$\scriptstyle L_1$}
    \move(0.04 0.04) \linewd 0.01 \setgray 0.7 \arrowheadtype t:F \arrowheadsize l:0.05 w:0.03 \lpatt()\avec(-0.23 0.04)
    \move(0.18 0.04) \avec(0.5 0.04)

    \htext(1.23 0.02){$\scriptstyle L_2$}
    \move(1.22 0.04) \avec(1 0.04)
    \move(1.37 0.04) \avec(1.61 0.04)

    \htext(2.53 0.02){$\scriptstyle L_3$}
    \move(2.52 0.04) \avec(2.1 0.04)
    \move(2.67 0.04) \avec(3.11 0.04)

    \move(2 -0.6)\htext{Figure 1.}
}

\section{Dirichlet-Neumann Bracketing}

\noindent Consider the partition of $\mathbb{R}^+$ defined by the set $P=\{0,x_1,x_2,\ldots\}$.
Let $H_D$ be the operator on $L^2(\mathbb{R}^+)$ defined by $H_D\psi=-\psi''+V_\omega\psi$ with
Dirichlet boundary conditions at each point of $P$, that is, $\psi(x_k)=0$. Similarly, define the
operator $H_N$ by $H_N\psi=-\psi''+V_\omega\psi$ with Neumann boundary conditions at each point of
$P$, that is, $\psi'(x_k)=0$. Then, both $H_D$ and $H_N$ are direct sums of operators on finite
intervals $I_k=[x_k,x_{k+1}]$.  Denote by $\mathcal{N}_{0,D}(W)$, the number of negative
eigenvalues for the operator $\mathcal{H}_D \equiv H_D-W$, and by $\mathcal{N}_{0,N}(W)$, the
number of negative eigenvalues for the operator $\mathcal{H}_N \equiv H_N-W$ (we will also
occasionally use the notations $\mathcal{N}_0(\mathcal{H}_D)$ and $\mathcal{N}_0(\mathcal{H}_N)$
for the numbers of negative eigenvalues of $\mathcal{H}_D$ and $\mathcal{H}_N$). The following
classical fact due to H. Weyl (so called Dirichlet-Neumann bracketing, see \cite{ReedSim}) will be
central in our calculations.

\begin{lemma} Let $H=H_0-W$ and $\mathcal{N}_0(W)=\#\{\lambda_i<0\}$.  Then
$\mathcal{N}_{0,D}(W) \leq \mathcal{N}_{0}(W) \leq \mathcal{N}_{0,N}(W)$.
\end{lemma}

\vskip0.4cm\noindent Let's take a closer look at the operators $\mathcal{H}_D$ and
$\mathcal{H}_N$. Since $H_D$ and $H_N$ are direct sums of operators on the intervals $I_k$, so are
$\mathcal{H}_D$ and $\mathcal{H}_N$. Denote these operators respectively by $\mathcal{H}_{D,k}$
and $\mathcal{H}_{N,k}$. Specifically, $\mathcal{H}_{D,k}$ and $\mathcal{H}_{N,k}$ are defined by
the same differential expression as $\mathcal{H}_D$ and $\mathcal{H}_N$, but with Dirichlet
boundary conditions at the edges of $I_k$ in the case of $\mathcal{H}_{D,k}$ and Neumann boundary
conditions in the case of $\mathcal{H}_{N,k}$. In this case, the numbers of negative eigenvalues
for $\mathcal{H}_D$ and $\mathcal{H}_N$ are the sum total of the numbers of negative eigenvalues
for the operators $\mathcal{H}_{D,k}$ and $\mathcal{H}_{N,k}$, that is,
$$\mathcal{N}_{0}(\mathcal{H}_D) = \sum_k \mathcal{N}_{0}(\mathcal{H}_{D,k}), \hskip0.5cm
\mathcal{N}_{0}(\mathcal{H}_N) = \sum_k \mathcal{N}_{0}(\mathcal{H}_{N,k}).$$ Now, the spectrum of
$\mathcal{H}_{D,k}$ and $\mathcal{H}_{N,k}$ are discrete for all $k$, with random eigenvalues
$\{\lambda^{(D)}_{n,k} \; : \; n \geq 0\}$ and $\{\lambda^{(N)}_{n,k} \; : \; n \geq 0\}$
accumulating only to $\infty$.

\vskip0.4cm\noindent The following definition will be important in what follows.  Two functions
$W^+(x)$ and $W^-(x)$ will be called upper and lower approximations for $W(x)$ if for some
$z_0>0$, $W^-(x) \leq W(x) \leq W^+(x)$ whenever $x \geq z_0$. We will construct them in the
following way.  Let us note that under the assumptions on $L_i$ we have
$x_k=L_1+\cdots+L_k+(2k-1)l$ for each $k$ and by the strong law of large numbers, $x_k/k \to
\alpha$ $P$-a.s. where $\alpha=\mathbb{E}[L_i]+2l.$ Therefore, given $\varepsilon>0$ there exists
$k_0=k_0(\varepsilon,\omega)$ such that when $k\geq k_0$, $(1-\varepsilon)\alpha k \leq x_k \leq
(1+\varepsilon) \alpha k$.  Now, for any $x \geq x_{k_0}$ there is some $k$ for which $x \in I_k$,
and since $W$ is monotonically decreasing,
$$W((1+\varepsilon)\alpha (k+1)) \leq W(x_{k+1}) \leq W(x) \leq
W(x_k) \leq W((1-\varepsilon) \alpha k).$$ So, for any $\varepsilon>0$, the functions
$W^+_\varepsilon$ and $W^-_\varepsilon$ defined as $W^+_\varepsilon(x) \equiv \sum^\infty_{k=1}
W((1-\varepsilon) \alpha k) \mathbb{I}_k(x)$ and $W^-_\varepsilon(x) \equiv \sum^\infty_{k=1}
W((1+\varepsilon) \alpha (k+1)) \mathbb{I}_k(x)$ are piecewise constant functions for which
$W^-_\varepsilon(x) \leq W(x) \leq W^+_\varepsilon(x)$ whenever $x \geq x_{k_0}$.  Thus,
$W^+_\varepsilon$ and $W^-_\varepsilon$ are upper and lower approximations for $W$.

\new By the Sturm comparison and oscillation theorems (see \cite{LevSarg}) it immediately follows that

\begin{equation} \mathcal{N}_{0,D,k}(W^-_\varepsilon) \leq
\mathcal{N}_{0,D,k}(W) \leq \mathcal{N}_{0,D,k}(W^+_\varepsilon)
\label{sturm1}
\end{equation}

\new when $k \geq k_0(\varepsilon,\omega).$ Further, since $W^-_\varepsilon, W^+_\varepsilon$ are
constant on intervals $I_k$, the spectra $\sigma ({H}_{\cdot,k}-W^-_\varepsilon(x))$ and $\sigma
({H}_{\cdot,k}-W^+_\varepsilon(x))$ can be determined exactly (see below).  We first consider a
simple example, the results of which will be used in the next section.

\begin{example} Let's consider the operators $h_N \psi=-\psi''+V(x)\psi$ and
$h_D \psi=-\psi''+V(x)\psi$ with Dirichlet boundary conditions at $-L-l$ and $L+l$ in the case of
$h_D$ and Neumann conditions in the case of $h_N$. Let $V$ be the symmetric function defined by
$V=h(\mathbb{I}_{[-L-l,-L]}+\mathbb{I}_{[L,L+l]})$: \vskip0.5cm

\noindent Both $h_N$ and $h_D$ have discrete spectrum bounded below by 0, with corresponding
eigenfunctions $\psi^{(D)}_i$ and $\psi^{(N)}_i$ for $i \geq 0$.  For even $i$, these
eigenfunctions are even, while odd for odd $i$. This is due to the fact that $V$ is symmetric. We
begin by estimating the lowest eigenvalue $\mu^{(D)}_0$ of $h_D$. Of course, the lowest eigenvalue
corresponds to the eigenfunction $\psi^{(D)}_0$. Set $\mu^{(D)}=k^2$. On the interval
$[-(L+l),-L]$, the eigenfunction $\psi^{(D)}(x)$ is given by
$$\psi^{(D)}(x)=A \sinh \sqrt{h-k^2}(x+L+l)$$ while on the interval $[-L,0]$,
$$\psi^{(D)}(x)=B\cos (kx).$$  The continuity conditions, $\psi^{(D)}(-L-0)=\psi^{(D)}(-L+0)$ and
$\psi'(-L-0)=\psi'(-L+0)$ immediately provide the eigenvalue equation,
\begin{equation}
k\tan(kL)=\sqrt{h-k^2}\coth(\sqrt{h-k^2}l). \label{evaleq}
\end{equation}
Looking ahead to the ``real" model, we solve the above equation for $L \to \infty$. Equation
(\ref{evaleq}) then implies $k_0 \to 0$ and that for some $\delta_0>0$
\begin{equation} k_0L=\frac{\pi}{2}-\delta_0 \label{evaleq2}
\end{equation}  The right side of (\ref{evaleq}) is analytic in $k$, and after expansion as a Taylor
series, we have

$$\tan(kL)=\frac{A_{0,D}}{k}+A_{1,D} k+O(k^3)$$

or equivalently
\begin{equation}
\cot(kL)=B_{0,D} k-B_{1,D}k^3+O(k^5) \label{evaleq3}
\end{equation}
where $A_{j,D}$ and $B_{j,D}$ are constants depending only on $h$ and $l$. In particular,
$B_{0,D}=1/A_{0,D}$ and $B_{1,D}=A_{1,D}/A_{0,D}$ where $A_{0,D}=\sqrt{h}\coth(\sqrt{h}l)$ and
$A_{1,D}>0$. Substitution of (\ref{evaleq2}) into (\ref{evaleq3}) gives
$$\tan(\delta_0)=B_{0,D} k_0-B_{1,D} k^3_0+O(k^5_0)$$ which for $\delta_0\to 0$ implies
$$\delta_0=B_{0,D} k_0-B_{1,D}k^3_0+O(k^4_0)$$ or in terms of $k_0$

\begin{equation} k_0=\frac{\pi}{2L}-\frac{1}{L}\left(B_{0,D} k_0-B_{1,D}k^3_0
\right)+O(k^4_0). \label{evaleq4}\end{equation} Iterating (\ref{evaleq4}) in $k_0$ gives

\begin{equation} \sqrt{\mu^{(D)}_0}=\frac{\pi}{2L}\left(1-\frac{B_{0,D}}{L}\right)
+O\left(\frac{1}{L^3}\right) \label{evaleq6}
\end{equation}

Similar steps give also the lowest eigenvalue for $h_N$.  In
fact, the eigenvalue equation for $h_N$ is
\begin{equation} k\tan(kL)=\sqrt{h-k^2}\tanh(\sqrt{h-k^2}l) \label{evaleq8}\end{equation} and gives
the lowest eigenvalue
\begin{equation} \sqrt{\mu^{(N)}_0}=\frac{\pi}{2L}\left(1-\frac{B_{0,N}}{L}\right)
+O\left(\frac{1}{L^3}\right) \label{evaleq7}
\end{equation} where $B_{0,N}$ is a positive constant depending only on $h$ and $l$.
\end{example}

\begin{remark}  It is interesting to note that the lowest eigenvalue for $h_{N}$ is asymptotically
equivalent to the lowest eigenvalue for $h_{D}$.
\end{remark}

\section{The Central Theorems}

\new The following results will be important in the study of the borderline decay of $W(x)$.

\new Given $\varepsilon>0$, set $w^+_{k,\varepsilon}=W((1- \varepsilon)\alpha k)$, $w^-_{k,\varepsilon}
=W((1+\varepsilon)\alpha k)$, and define the upper and lower approximations
$W^+_\varepsilon(x)=\sum^\infty_{k=1} w^+_{k,\varepsilon} \mathbb{I}_k(x)$ and
$W^-_\varepsilon(x)=\sum^\infty_{k=1} w^-_{k,\varepsilon} \mathbb{I}_k(x)$ Where $\mathbb{I}_k$
are the indicator functions of the intervals $I_k=[x_k,x_{k+1}]$.  Define the operators
$H_{D,k}=-d^2/dx^2+V_k(x)$ with Dirichlet boundary conditions at $x_k$ and $x_{k+1}$ and
$H_{N,k}=-d^2/dx^2+V_k(x)$ with Neumann boundary conditions at $x_k$ and $x_{k+1}$. Here
$V_k=\mathbb{I}_{[x_k,x_k+l]}+\mathbb{I}_{[x_{k+1}-l,x_{k+1}]}$. Set
$\mathcal{H}_{D,k,\varepsilon}=H_{D,k}-W^-_\varepsilon \mathbb{I}_k$ and
$\mathcal{H}_{N,k,\varepsilon}=H_{N,k}-W^+_\varepsilon \mathbb{I}_k$. Finally $\mu^{(D)}_{0,k}$
and $\mu^{(N)}_{0,k}$ denote respectively, the lowest eigenvalues of $H_{D,k}$ and $H_{N,k}$.

\begin{theorem} \label{mainthm1} Let $\varepsilon>0$ be given.   If
\begin{equation}\sum^\infty_{k=1} P\left(\mu^{(N)}_{0,k} <
w^+_{k,\varepsilon}\right)<\infty \label{mainthm1inequality}
\end{equation} then $P$-a.s. $\mathcal{N}_0(W)<\infty.$
\end{theorem}

\begin{proof} By the strong law of large numbers, there is some $k_0=k_0(\omega,\varepsilon)$ so that
if $k \geq k_0$ and $x\in I_k$, then $(1-\varepsilon) \alpha k \leq x$ in which case $W(x) \leq
W((1-\varepsilon)\alpha k)$.  For such $k$, the Sturm oscillation and comparison theorems imply
that $\mathcal{N}_{0,k,N}(W) \leq \mathcal{N}_{0,k,N}(W^+_{\varepsilon}\mathbb{I}_k)$, i.e., the
number of negative eigenvalues for the operator $H_{N,k}-W$ is not greater than the number of
negative eigenvalues of $H_{N,k}-W^+_\varepsilon \mathbb{I}_k$.  The lowest eigenvalue for
$H_{N,k}-W^+_\varepsilon \mathbb{I}_k$ is $\mu^{(N)}_{0,k}-w^+_{k,\varepsilon}$ since
$W^+_\varepsilon(x)=w^+_{k,\varepsilon}$ is constant for all $x \in I_k$.  If
(\ref{mainthm1inequality}) holds, then the Borel-Cantelli lemma implies that
$\mu^{(N)}_{0,k}-w^+_{k,\varepsilon}\geq 0$ $P$-a.s. for $k \geq k_1(\omega,\varepsilon) \geq
k_0$. It means $\mathcal{N}_{0,k,N}(W^+_\varepsilon \mathbb{I}_k)=0$ and consequently
$\mathcal{N}_{0,k,N}(W)=0$ for all $k \geq k_1$. Dirichlet-Neumann bracketing implies
$\mathcal{N}_0(W) \leq \mathcal{N}_{0,N}(W)=\sum \mathcal{N}_{0,k,N}(W)<\infty.$
\end{proof}

\new Independence of the events $\{\mu^{(D)}_{0,k} <
w^-_{k,\varepsilon}\}$ and similar ideas lead to the following theorem:

\begin{theorem} Let $\varepsilon>0$ be given.   If
$$\sum^\infty_{k=1} P\left(\mu^{(D)}_{0,k} <
w^-_{k,\varepsilon}\right)=\infty$$ then $P$-a.s. $\mathcal{N}_0(W)=\infty.$
\end{theorem}

\new In connection with the above theorems, we have the following lemmas:

\begin{lemma} \label{mainlemma1} For any $\varepsilon>0$ let $w^+_{k,\varepsilon}=
W((1-\varepsilon)\alpha k)$.  There exists a positive constant $c_N$ depending only on $l$ and $h$
such that $\mathcal{N}_0(W)<\infty$ $P$-a.s. whenever
\begin{equation}\sum^\infty_{k=1} P\left(L_k >
\pi/\sqrt{w^+_{k,\varepsilon}}-c_N\right)<\infty \label{mainlem1eqn}
\end{equation}
\end{lemma}

\begin{proof} Suppose that $P$-a.s $N_0(W)=\infty$. Then, by Theorem (\ref{mainthm1}) $$\sum_k
P(\mu^{(N)}_{0,k} < w^+_{k,\varepsilon})=\infty.$$ Since $\mu^{(N)}_{0,k}$ are independent random
variables the second Borel-Cantelli lemma implies that $P$-a.s, there is a collection of intervals
$I_k$ with $k \to \infty$ on which $\sqrt{\mu^{(N)}_{0,k}} < \sqrt{w^+_{k,\varepsilon}}$.
Therefore, $\sqrt{\mu^{(N)}_{0,k}} \to 0$ and by (\ref{evaleq8}) it follows that $L_k \to \infty$.
In this case by the asymptotic formula (\ref{evaleq7})
$$\frac{\pi}{L_k}\left(1-\frac{B_{0,N}}{2L_k}\right) < \sqrt{w^+_{k,\varepsilon}}$$ for all $k$
sufficiently large. Recall that $B_{0,N}$ is a positive constant depending only on $l$ and $h$.
This inequality and the fact that $L_k \to \infty$ immediately give $$L_k
\sqrt{w^+_{k,\varepsilon}} \geq \frac{1}{2}$$ for $k>>1$. Consequently, $$L_k >
\frac{\pi}{\sqrt{w^+_{k,\varepsilon}}}-\frac{\pi B_{0,N}}{2\sqrt{w^+_{k,\varepsilon}}L_k} \geq
\frac{\pi}{\sqrt{w^+_{k,\varepsilon}}}-c_N$$ for sufficiently large $k$ where $c_N$ is a positive
constant depending only on $l$ and $h$.  In this case $$P(\mu^{(N)}_{0,k} <
w^+_{k,\varepsilon})=P\left(L_k> \pi/\sqrt{w^+_{k,\varepsilon}}-c_N\right)$$ so that $\ds \sum_k
P\left(L_k> \pi/\sqrt{w^+_{k,\varepsilon}}-c_N\right)=\infty$.
\end{proof}
\new Similar steps lead to the proof of the following lemma:

\begin{lemma} \label{mainlemma2} For any $\varepsilon>0$ let
$w^-_{k,\varepsilon}=W((1+\varepsilon)\alpha k)$.  There exists a positive constant $c_D$
depending only on $l$ and $h$ such that $\mathcal{N}_0(W)=\infty$ $P$-a.s. whenever
$$\sum^\infty_{k=1} P\left(L_k > \pi/\sqrt{w^-_{k,\varepsilon}}-c_D\right)=\infty$$
\end{lemma}

\begin{remark} It follows that Theorems 1.3 and 1.4 presented in the introduction are a consequence
of Lemmas 4.1 and 4.2 above. In fact, if $P(L_k>x)\sim \exp(-\eta x^\alpha/\alpha)$, then for
$c=c_D$ or $c=c_N$, $P\left(L_k > \pi/\sqrt{w^-_{k,\varepsilon}}-c\right) \sim CP\left(L_k >
\pi/\sqrt{w^-_{k,\varepsilon}}\right)$ for some positive constant $C$.  The borderline for
convergence or divergence of the sums $\sum P\left(L_k> \pi/\sqrt{w^-_{k,\varepsilon}}-c\right)$
is then achieved by taking $W(x) \sim C_{\eta,\alpha}/\ln^{2/\alpha} x$ where
$C_{\eta,\alpha}=(\eta \alpha^{-1})^{2/\alpha} \pi^2$.  In a similar way, if $P(L_k>x) \sim
x^{-\alpha}$, then $P\left(L_k > \pi/\sqrt{w^-_{k,\varepsilon}}-c\right) \sim CP\left(L_k >
\pi/\sqrt{w^-_{k,\varepsilon}}\right)$ for some positive constant $C$ and the borderline between
convergence or divergence is achieved if $W(x) \sim x^{-2/\alpha}$.
\end{remark}

\section{Examples}

\new Our goal in this section is the proof of the following theorem.

\begin{theorem} \label{mainexample1} Let $\{L_k\}$ be a sequence of independent, identically
distributed non-negative random variables having exponential tails, that is, $$P\{L_k>x\} \sim
e^{-\eta x}$$ for some $\eta>0$.  Let $c_0=\eta^2\pi^2$.  If $W(x) < \ds c_0/\ln^2 x$ for all
$x>>1$, then $\mathcal{N}_0(W)<\infty$ $P$-a.s. If $W(x) > \ds c_0/\ln^2 x$ for $x>>1$ then
$\mathcal{N}_0(W)=\infty$ $P$-a.s.
\end{theorem}

\begin{proof}  Suppose $W(x)\leq (c_0(1-\delta))/\ln^2 x$ for some $0<\delta<1$ and $x>>1$.
Then $$P\{L_k > \pi/\sqrt{w^+_{k,\varepsilon}}-c_N\} \sim \frac{c}{k^{1+\beta}}$$ where
$c=c(\eta,\varepsilon,\delta)$ and $1+\beta=(1-\delta)^{-1}$.  Consequently,
$$\sum^\infty_{k=1} P\{L_k>\pi/\sqrt{w^+_{k,\varepsilon}}-c_N\}<\infty$$ and therefore by
Lemma \ref{mainlemma1}, $\mathcal{N}_0(W)<\infty$ $P$-a.s.  On the other hand, if $W(x) \geq
(c_0(1+\delta))/\ln^2 x$ for $x>>1$, then $$P\{L_k > \pi/\sqrt{w^-_{k,\varepsilon}}-c_D\} \geq
\frac{c}{k^{1-\beta}}$$ where $1-\beta=(1+\delta)^{-1}$ and therefore $$\sum^\infty_{k=1} P\{L_k
>\pi/\sqrt{w^-_{k,\varepsilon}}-c_D\}=\infty.$$  By Lemma
\ref{mainlemma2}, $\mathcal{N}_0(W)=\infty$ $P$-a.s.
\end{proof}

\begin{example} Consider the Bernoulli potential $V_\omega(x)=h\sum \epsilon_k \xi_k(x)$ described
in the introduction. By the Borel-Cantelli lemma, it follows that there exists $P$-a.s a sequence
$x_{n_k}$ of positive integers for which $V(x,\omega)=0$ on $[x_k+1/2,x_{k+1}-1/2)$ and
$L_k=x_{k+1}-x_k-1 \to \infty$.  Since for any $m \in \mathbb{Z}^+$, $P(L_k \geq
m)=q^m=e^{-m\ln(1/q)}$ it follows that $L_k$ have exponential tails.  Hence by Theorem
\ref{mainexample1}, the borderline between finite and infinite negative eigenvalues is given by a
decay in $W(x)$ of the order $c_1/\ln^2{x}$ with $c_1=\pi^2\ln^2(1/q)$.
\end{example}

\bigskip

\clearpage
\end{document}